\documentclass[12pt]{article}

\usepackage[cp1251]{inputenc}

\usepackage{amsfonts}
 \usepackage{amsthm}
 \usepackage{amsmath,bm}
 \usepackage{amscd}
 \usepackage{amssymb}

\usepackage[english]{babel}

 \usepackage{graphics}
\usepackage{graphicx}

 \newtheorem{thm}{Theorem}
 
 \newtheorem{prop}{Proposition}
 \theoremstyle{definition}

\title{A note on the integrability status of a 'mysterious' case of a quadratic Hamiltonian}
\author{Ognyan Christov\\
Faculty of Mathematics and Informatics, Sofia University, \\
5 J. Bouchier blvd., 1164 Sofia, Bulgaria}
\date{}

\begin{document}

 \maketitle

\begin{abstract}
\noindent
We present a simple algebraic proof of meromorphic
non-integrability of a 'mysterious' case of a quadratic
Hamiltonian studied by Sokolov and Wolf. The result is obtained by
applying the Morales-Ramis method, based on differential Galois
theory.
\end{abstract}



\section{Introduction and main result}
\label{intro}

Sokolov and Wolf studied the integrability of a special class of
quadratic Hamiltonians describing various physical models
\cite{SW}
\begin{equation}
\label{1.1}
H = c_1 (\mathbf{a}, \mathbf{b}) |\mathbf{M} |^2 + c_2
(\mathbf{a}, \mathbf{M}) (\mathbf{b}, \mathbf{M}) + (\mathbf{b},
\mathbf{M}\times \mathbf{\Gamma}),
\end{equation}
where $c_1, c_2$, $\mathbf{a} = (a_1, a_2, a_3)$ and $\mathbf{b} =
(b_1, b_2, b_3)$ are parameters, $\mathbf{M} = (M_1, M_2, M_3)$
and $\Gamma = (\gamma_1, \gamma_2, \gamma_3)$.

Together with the above Hamiltonians, they consider the following
family of Poisson brackets
\begin{equation}
\label{1.2}
\{M_i, M_j \} = \varepsilon_{ijk} M_k, \quad \{M_i,
\gamma_j \} = \varepsilon_{ijk} \gamma_k, \quad \{\gamma_i,
\gamma_j \} = K \varepsilon_{ijk} M_k .
\end{equation}
Here $K$ is a parameter and $\varepsilon_{ijk}$ is the totally
skew-symmetric tensor. The cases $K=0$, $K>0$ and $K<0$ correspond
to the Lie algebras $e (3)$, $so (4)$ and $so (3, 1)$.

The bracket (\ref{1.2}) admits two Casimirs
\begin{equation}
\label{1.3}
J_1 = (\mathbf{M}, \mathbf{\Gamma}), \qquad J_2 = K |\mathbf{M} |^2 + |\mathbf{\Gamma} |^2.
\end{equation}
Therefore, in order to achieve  complete integrability, one more
first integral that is functionally independent of the Hamiltonian
(\ref{1.1}) and the Casimirs (\ref{1.3}) is needed. Notice also
that the Hamiltonian and the bracket are invariant with respect to
the transformations
\begin{equation}
\label{1.4}
\mathbf{M} \to T \mathbf{M}, \qquad \mathbf{\Gamma} \to T \mathbf{\Gamma},
\end{equation}
where $T$ is a constant orthogonal matrix. As a result of such a
transformation, we can achieve that the vectors $\mathbf{a},
\mathbf{b}$ become $\mathbf{a} = (a_1, 0, a_3)$ and $\mathbf{b} =
(0, 0, 1)$.

Using the Kowalewski-Lyapounov test, Sokolov and Wolf found the
parameter values for which integrability is possible. These values
are
\begin{enumerate}
\item [(1)] $c_1$ - arbitrary, $c_2 = 0$,

\item [(2)] $c_1 = 1$, $c_2 = -2$,

\item [(3)] $c_1 = 1$, $c_2 = -1$,

\item [(4)] $c_1 = 1$, $c_2 = -1/2$,

\item [(5)] $c_1 = 1$, $c_2 = 1$.
\end{enumerate}
Indeed, they succeeded in finding polynomial first integrals in
all these cases except (5), for which there are no polynomial
first integrals up to degree 8; this is why they called it
'mysterious'.

With the parameters $c_1, c_2$ as in case (5), the Hamiltonian
(\ref{1.1}) takes the form
\begin{equation}
\label{1.5}
H = a_3 (M_1^2 + M_2^2) + 2 a_3 M_3^2 + a_1 M_1 M_3 + M_1 \gamma_2 - M_2 \gamma_1.
\end{equation}
As suggested by Sokolov and Wolf, we investigate the integrability
of this Hamiltonian when the Poisson bracket parameter $K < 0$
(that is, the Poisson bracket has the type $so (3,1)$) is related
to the length of $\mathbf{a}$ by the relation
\begin{equation}
\label{1.6}
a_1 ^2 + a_3 ^2 = - K.
\end{equation}

The study of the integrability makes sense since the Hamiltonian
(\ref{1.5}) still depends on the parameters $a_1, a_3$. Sakovich
\cite{Sakovich} has found that this case does not pass the
Painlev\'{e} test when $a_1, a_3 \neq 0$. However, when $a_1 = 0$
there is an additional first integral $M_3 = const$. When $a_3 =
0$ the Hamiltonian reduces to the case (3) above, in which an
additional polynomial first integral of degree four exists.

In this short note we give an alternative proof of the Sakovich
result:
\begin{thm}
\label{Th1} The Hamiltonian (\ref{1.5}) does not admit an
additional meromorphic first integral when $a_1, a_3 \neq 0$, that
is, the system corresponding to (\ref{1.5}) is non-integrable.
\end{thm}

This result also answers  the question of the integrability of a
Hamiltonian, equivalent to (\ref{1.5}), but with a different
normalization, raised in the recent review article \cite{Sokolov}.

We use the Ziglin-Morales-Ramis theory for integrability of
Hamiltonian systems based on differential Galois theory. The
application of the Ziglin-Morales-Ramis theory requires
 a non-equilibrium solution to the corresponding Hamiltonian
system. Then, the Galois group of the variational or normal
variational equations along this particular solution has to be
studied. For the Hamiltonian system under consideration, we
succeeded in finding the solutions of the normal variational
equations in quadratures. After that,  determining the Galois group is
easy. Similar calculations of the differential Galois group were
performed in \cite{MSS,OC,OC1}. It should be noted that this is
not the case in general: finding the differential Galois group is
a difficult task.

The outline of this note is as follows. In Section 2 we summarize,
for the reader's convenience, some facts about differential Galois
theory and Ziglin-Morales-Ramis theory about Hamiltonian
integrability. The proof of Theorem \ref{Th1} is carried out in
Section 3.

\section{Preliminaries}
\label{prel}

In this section, we briefly recall some notions and results
related to Ziglin-Morales-Ramis theory, which turns out to be very useful
when studying integrability in the complex domain.

To begin with, let us consider a Hamiltonian system
\begin{equation}
\label{2.1}
\dot{\mathbf{x}} = X_{H} (\mathbf{x}), \quad t \in \mathbb{C}, \quad \mathbf{x} \in M
\end{equation}
corresponding to an analytic Hamiltonian $H$, defined on the
complex $2 n$-dimensional manifold $M$. Similarly to the real case
\cite{AKN}, we call such a Hamiltonian system  completely integrable if we are given $n$  independent (almost
everywhere) first integrals in involution.

Assume that we can find  a non-equilibrium solution $\Phi (t)$ for the system (\ref{2.1}) and let
 $\mathbb{G}$ be its phase curve. Along this solution we can write the variational equations (VE)
\begin{equation}
\label{2.2}
\dot{\boldsymbol\xi} = D X_{H} ( \Phi (t)) \boldsymbol\xi, \quad \boldsymbol\xi \in T_{\mathbb{G}} M.
\end{equation}
The first integral $H$ gives rise to a linear integral $d H$ of
the variational equations. Using the integral $d H$ we can reduce
the variational equations. Consider the normal bundle of $
\mathbb{G}$, $F:= T_{\mathbb{G}} M / TM$  and let $\pi :
T_{\mathbb{G}} M \to F$ be the natural projection. The system of
equations (\ref{2.2}) defines a system of equations on $F$
\begin{equation}
\label{2.3}
\dot{\boldsymbol\eta} = \pi_{*} (D X_{H} ( \Phi (t))(\pi^{-1} \boldsymbol\eta) , \quad \boldsymbol\eta \in F.
\end{equation}
These are called the normal variational equations (NVE). Each
meromorphic first integral of the Hamiltonian system (\ref{2.1})
in the neighborhood of the curve $ \mathbb{G}$ gives rise to a
meromorphic first integral of (NVE) \cite{Ziglin}. Hence, the
problem of complete integrability of the Hamiltonian system
(\ref{2.1}) reduces to the study of integrability of the linear
system (\ref{2.3}) (or (\ref{2.2})).

Consider such a linear non-autonomous system
\begin{equation}
\label{2.4}
\dot{\boldsymbol\xi} = P (t) \boldsymbol\xi, \quad \boldsymbol\xi \in \mathbb{C}^n ,
\end{equation}
with $t$ defined on some Riemann surface $\mathbb{G}$. By the
existence theorem  there is a fundamental matrix solution $\Psi
(t)$, analytic in a vicinity of any nonsingular point $t_0$.

Further, we recollect  briefly certain notions and results which reveal the relationship between the
differential Galois theory and the integrability of Hamiltonian systems. All the necessary background and technical details
can be found in \cite{vPS,M}.

A diffe\-ren\-tial field $\mathbb{K}$ is a field with a derivation
$\partial = '$, i.e. an additive mapping satisfying the Leibnitz
rule. A differential automorphism of $\mathbb{K}$ is an
auto\-mor\-phism commuting with the derivation.

Denote the coefficient field in (\ref{2.4}) by $\mathbb{K}$.
If $\Psi_{i j}$ are the elements of the fundamental matrix $\Psi
(t)$ we construct an extension $\mathbb{F} (\Psi_{i j})$
 generated by $\mathbb{K}$ and $\Psi_{i j}$, which is
a differential field. This extension is called a Picard-Vessiot
extension. The Galois group $G :=Gal (\mathbb{F}/\mathbb{K})$ is
defined to be the group of all differential automorphisms of
$\mathbb{F}$ leaving the elements of $\mathbb{K}$ fixed. The
Galois group is an algebraic group. It has a unique connected
component $G^0$ which contains the identity and which is a normal
subgroup of finite index. The Galois group $G$ can be represented
as an algebraic li\-near subgroup of $\mathrm{GL} (n, \mathbb{C})$
by
$$
\sigma \Psi (t) = \Psi (t) R_{\sigma},
$$
where $\sigma \in G$ and $R_{\sigma} \in \mathrm{GL} (n,
\mathbb{C})$.

 The fundamental result of the Morales-Ramis theory is then the following
 \begin{thm}
 \label{thMR}
 ({\rm Morales-Ruiz and  Ramis \cite{M}}) Suppose that a Hamiltonian system has $n$
  meromorphic first integrals in involution. Then, the identity component  $G^0$
  of the Galois group $G = Gal(\mathbb{F}/\mathbb{K})$ is abelian.
 \end{thm}
Notice that this result is only a necessary condition. Proving
that the identity component of the Galois group is abelian does
not guarantee that the corresponding Hamiltonian system is
integrable.

Examples of the application of this theorem to prove the
non-integrabi\-li\-ty of important Hamiltonian systems can be
found, for instance, in \cite{M,MSS,OC,OC1}.

The difficult part when applying the Morales-Ramis theorem is to decide whether differential Galois group of VE (NVE) is abelian or not.
It is worth mentioning some algorithmic procedures providing information on the differential Galois group for various equations and systems.
The Kovacic algorithm is nowadays a classical result which gives the fundamental system of solution
of second order differential equation  in the event that the equation is integrable, and nothing if it is not
(see, e.g. \cite{Kovacic,M}).
In recent years, methods have been developed \cite{ADW} that utilize the structure of the  system of the variational equations
to make it simpler, which will allow one to check efficiently whether its Lie algebra is abelian or not.

Finally, we recall another important notion which we will use later.
Let $\mathbb{K}$ be a differential field with an algebraically
closed subfield of constants. An extension $\mathbb{L} /
\mathbb{K}$ is called a Liouville extension of $\mathbb{K}$ if
$Const (\mathbb{K}) = Const (\mathbb{L})$ and there exists a chain of
differential extensions
$$
\mathbb{K} = \mathbb{L}_0 \subset \mathbb{L}_1 \subset \ldots \subset \mathbb{L}_n = \mathbb{L}
$$
such that for $i = 1, \ldots , n$ $\mathbb{L}_i = \mathbb{L}_{i-1}
(\alpha_i) $ and $\alpha_i$ satisfies one of the following conditions:

1) $\alpha_i ' \in \mathbb{L}_{i-1}$: we say that $\alpha_i$ is an
integral of an element of $\mathbb{L}_{i-1}$;

2) $\alpha_i \neq 0$ $\alpha_i '/\alpha_i \in \mathbb{L}_{i-1}$:
we say $\alpha_i$ is an exponential of an integral of an element
of $L_{i-1}$;

3) $\alpha_i$ is algebraic over $\mathbb{L}_{i-1}$.

It can be proven that $\mathbb{L}$ is a Liouville extension of
$\mathbb{K}$ if, and only if, the identity component $G^0$ of $G =
Gal (\mathbb{L} / \mathbb{K})$ is a solvable subgroup.

\section{Proof of Theorem \ref{Th1}}
\label{proofth1}

First, we write the equations of motion corresponding to the
Hamiltonian (\ref{1.5}):
\begin{align}
\label{3.1}
\dot{M}_1 &= 2 a_3 M_2 M_3 + a_1 M_1 M_2 + M_3 \gamma_1 - M_1 \gamma_3, \nonumber \\
\dot{M}_2 &= -2 a_3 M_1 M_3 + a_1 (M_3^2 -  M_1^2)  + M_3 \gamma_2 - M_2 \gamma_3, \nonumber \\
\dot{M}_3 &= - a_1 M_2 M_3 , \\
\dot{\gamma}_1 &= 4 a_3 M_3 \gamma_2 - 2 a_3 M_2 \gamma_3  + a_1 M_1 \gamma_2 - K M_1 M_3 + \gamma_1 \gamma_3, \nonumber \\
\dot{\gamma}_2 &= 2 a_3 M_1 \gamma_3 - 4 a_3 M_3 \gamma_1  + a_1 (M_3 \gamma_3 - M_1 \gamma_1)  - K M_2 M_3 + \gamma_2 \gamma_3, \nonumber \\
\dot{\gamma}_3 &= 2 a_3 (M_2 \gamma_1 - M_1 \gamma_2) - a_1 M_3
\gamma_2 + K (M_1 ^2 + M_2 ^2) - \gamma_1 ^2 - \gamma_2 ^2.
\nonumber
\end{align}
Evidently, $ a_1 = 0$ gives the first integral $M_3 = const$.

To prove the theorem, we need to find a non-equilibrium solution,
to write down the normal variational equations along this
particular solution and to show that their differential Galois
group is not abelian.

In what follows, we assume  $a_1, a_3 \neq 0$. We also keep $a_1,
a_3 \in \mathbb{R}$, but consider $t, M_j, \gamma_j, j = 1, 2, 3$
to be complex quantities.

\begin{prop}
Let $\alpha$ be an arbitrary real number such that $0 < \alpha^2 <
a_1 ^2$ and $\frac{\alpha^2}{a_1 ^2} \notin \mathbb{Q}$. Then, the
system (\ref{3.1}) admits the following particular solution
\begin{align}
\label{3.2}
\mathbb{G}:   & M_1 = M_3 = \gamma_2 = 0, \\
              & \quad M_2 = sn (a_1 t, \kappa), \quad \gamma_1 = \frac{i \alpha}{M_2} + a_3 M_2, \quad \gamma_3 = - \frac{\dot{M}_2}{M_2}, \nonumber
\end{align}
where $sn$ is the Jacobi elliptic function with a modulus $\kappa^2 =  \alpha ^2 /a_1 ^2$ (see, e.g. \cite{WW}).
\end{prop}
The proof is immediate.

\vspace{2ex}

Denote $\xi_j = d M_j, \eta_j = d \gamma_j, j = 1, 2, 3$. Then,
the normal variational equations (NVE) along the solution
(\ref{3.2}) are
\begin{align}
\label{3.4}
\dot{\eta}_2 &= - \frac{\dot{M}_2}{M_2} \eta_2 - \left(2 a_3 \frac{\dot{M}_2}{M_2} + \frac{ i a_1 \alpha}{M_2} + a_1 a_3 M_2 \right) \xi_1 \nonumber\\
             & \quad - \left( a_1 \frac{\dot{M}_2}{M_2} + \frac{4 i a_3 \alpha}{M_2}   + (4 a_3^2 + K) M_2 \right) \xi_3 , \nonumber \\
\dot{\xi}_1  &= \left(a_1 M_2 + \frac{\dot{M}_2}{M_2} \right) \xi_1 + \left(3 a_3 M_2 + \frac{i \alpha}{M_2} \right) \xi_3, \\
\dot{\xi}_3  &= -a_1 M_2 \xi_3. \nonumber
\end{align}
This system has a convenient triangular structure, so  we can
disregard the first equation of (\ref{3.4}) for now while
considering the subsystem
\begin{align}
\label{3.5}
\dot{\xi}_1  &= \left(a_1 M_2 + \frac{\dot{M}_2}{M_2} \right) \xi_1 + \left(3 a_3 M_2 + \frac{D_1}{M_2} \right) \xi_3, \nonumber \\
\dot{\xi}_3  &= -a_1 M_2 \xi_3
\end{align}
and studying its differential Galois group, which is sufficient
for our purposes. This reduction is justified, because from  the first integral of (\ref{3.4})
$$
d J_1 = \gamma_1 \xi_1 + M_2 \eta_2 + \gamma_3 \xi_3 = const,
$$
which stems from the Casimir $J_1$, we can always recover $\eta_2$.

Since the equations are linear, we obtain
\begin{align}
\label{3.6}
\xi_3 &= \exp (-a_1 \int M_2 dt) \xi_3 ^0,  \nonumber \\
\xi_1 &= M_2 \exp (a_1 \int M_2 dt) \xi_1 ^0 \\
      &+ M_2 \exp (a_1 \int M_2 dt) \int \left( \frac{i \alpha}{M_2 ^2} + 3 a_3\right) \exp (-2 a_1 \int M_2 dt) d t \xi_3 ^0 . \nonumber
\end{align}
Introduce $I:= \exp (a_1 \int M_2 dt)$ and
\begin{equation}
\label{3.7}
\Psi_{33} :=I^{-1}, \quad \Psi_{22} := M_2 I, \quad
\Psi_{23} := \Psi_{22} \int \left( \frac{i \alpha}{M_2 ^2} + 3 a_3 \right) I^{-2} d t .
\end{equation}
Then,
\begin{equation}
\label{3.8}
\xi_1 = \Psi_{22} \xi_1 ^0 + \Psi_{23} \xi_3 ^0, \qquad \xi_3 = \Psi_{33} \xi_3 ^0.
\end{equation}

Thus, the fundamental matrix of (\ref{3.5}) is
\begin{equation}
\label{3.9}
\Psi =
\begin{pmatrix}
 \Psi_{22} & \Psi_{23} \\
 0         & \Psi_{33}
\end{pmatrix}.
\end{equation}

Let the coefficient field $\mathbb{K}$ for the (NVE) and for the
(\ref{3.5}) be the field of elliptic functions with the modulus $\kappa$.

We now  show that  $I$ does not belong to the coefficient field $\mathbb{K}$. Recall that
$$
\int M_2 dt = \int sn (a_1 t, \kappa) dt = \frac{1}{a_1} \ln
\left( dn (a_1 t, \kappa) - \kappa cn (a_1 t, \kappa)
\right)^{\frac{1}{\kappa}},
$$
where $dn, cn$ are the other basic Jacobi elliptic functions. This
integral clearly is not an elliptic function, that is, $\int M_2 dt \notin \mathbb{K}$.

Furthermore,
$$
I = \exp (a_1 \int M_2 dt) = \exp ( \ln \left( dn (a_1 t, \kappa) - \kappa cn (a_1 t, \kappa)  \right)^{\frac{1}{\kappa}})
  = \left( dn (a_1 t, \kappa) - \kappa cn (a_1 t, \kappa)  \right)^{\frac{1}{\kappa}},
$$
which again is not an  elliptic function since $\kappa \notin \mathbb{Q}$.

By adjoining the element $I$ to the field $\mathbb{K}$, we get the following Picard-Vessiot extension $\mathbb{F}_1 : = \mathbb{K} (I)$.
Let $\sigma \in Gal (\mathbb{F}_1 /\mathbb{K})$. Then, $ \sigma (\int M_2 dt) = \int M_2 dt + \nu, \nu \neq 0 $
and $\sigma (I) = c I, c \neq 0, 1$, and therefore
$$
\sigma (\Psi_{33}) = c^{-1} \Psi_{33}.
$$
Clearly, $Gal (\mathbb{F}_1 /\mathbb{K}) \cong \mathbb{C}^*$.
From (\ref{3.7}) we also see that $\Psi_{22} \in \mathbb{F}_1 $, so $\sigma (\Psi_{22}) = c \Psi_{22}$.

It remains to determine whether  $\Psi_{23}$ belongs to $\mathbb{F}_1 $ and more precisely, whether the quadrature $J$ belongs to $\mathbb{F}_1 $
$$
J:= \int \left( \frac{i \alpha}{M_2 ^2} + 3 a_3 \right) I^{-2} d t .
$$
\begin{prop}
\label{prop2}
$J$ does not belong to $\mathbb{F}_1 = \mathbb{K} (I)$.
\end{prop}
\begin{proof}
Suppose, to the contrary,  that   $J \in \mathbb{F}_1 = \mathbb{K} (I)$.
Evidently,
\begin{equation}
\label{3.10}
\dot{J} = \left( \frac{i \alpha}{M_2 ^2} + 3 a_3 \right) I^{-2} \in \mathbb{F}_1 .
\end{equation}
Since $\sigma (I^2 J) = I^2 J $, there exists $V \in \mathbb{K}$, such that
\begin{equation}
\label{3.11}
J = V I^{-2} .
\end{equation}
Differentiating (\ref{3.11}) and combining it with (\ref{3.10}) yields
\begin{equation}
\label{3.12}
\dot{V} - 2 a_1 M_2 V = \frac{i \alpha}{M_2 ^2} + 3 a_3 .
\end{equation}
We will show that this equation has no solution in $\mathbb{K}$.

Recall that $M_2 = sn (a_1 t, \kappa)$ is an odd function and it has two simple poles and two simple zeroes in the parallelogram of the periods.
Suppose that $t_0$ is a zero of $M_2$. Then, near $t_0$ we have
$$
M_2 = \pm a_1 (t-t_0) + O ((t-t_0)^3).
$$
If $V$ is a meromorphic solution of (\ref{3.12}), it should have the expansion
$$
V = \frac{v_{-1}}{t-t_0} + v_0 + v_1 (t-t_0) + \ldots .
$$
Comparing the coefficients in front of $(t-t_0)^{-2}$ on the both sides of (\ref{3.12}) gives $v_{-1} = - i \alpha /a_1 ^2$. Hence,
\begin{equation}
\label{3.13}
{\rm Res}_{t = t_0} V = - \frac{i \alpha}{a_1 ^2}.
\end{equation}
Let $t_p$ be a pole of $M_2$ and denote $\tau = t - t_p$. We have
$$
M_2 = \frac{\rho}{\tau} + O (\tau), \quad \rho = \pm \frac{1}{a_1 \kappa}.
$$
Suppose $V$ has a pole at $t_p$
$$
V = r_m \tau ^m + \ldots , \quad m \in \mathbb{Z}_{-}, \quad r_m \neq 0.
$$
Then,
$$
\dot{V} -2 a_1 M_2 V = \left( m -2 a_1 \rho \right) r_m \tau^{m-1} + \ldots.
$$
The right-hand side of (\ref{3.12}) is regular at $t_p$.
In order the left-hand side to be regular at $t_p$, one needs
$$
m = 2a_1 \rho = \pm \frac{2}{\kappa},
$$
which is impossible, because $m \in \mathbb{Z}_{-}$, while $\kappa \notin \mathbb{Q}$.

It is also clear that $V$ does not have poles at the regular points of $M_2$. So, the only possible poles of $V$ are in the two
simple zeroes of $M_2$. It follows from (\ref{3.13}) that
$$
\sum {\rm Res} V = - \frac{2 i \alpha}{a_1 ^2} \neq 0,
$$
which is a contradiction to the theorem about the sum of the residues of an elliptic function. Therefore, (\ref{3.12}) has no
solution $V \in \mathbb{K}$, hence the claim follows.
\end{proof}

Adjoining the element $J$, we get another Picard-Vessiot extension
\begin{equation}
\mathbb{K} \subset \mathbb{F}_1 \subset \mathbb{F}_2 := \mathbb{F}_1 (J).
\end{equation}
The extension  $\mathbb{F}_2$ is a Liouville extension.
Therefore, $G = Gal (\mathbb{F}_2/\mathbb{K})$ is a sol\-va\-ble group and so is its identity component $G^0$.
If $\sigma \in Gal (\mathbb{F}_2/\mathbb{K})$ we have
\begin{equation}
\label{3.14}
\sigma \Psi =
\begin{pmatrix}
 \sigma(\Psi_{22}) & \sigma(\Psi_{23}) \\
 0                 & \sigma(\Psi_{33})
\end{pmatrix} = \Psi R_{\sigma}.
\end{equation}
Clearly, the matrix $R_{\sigma}$ takes the form
\begin{equation}
\label{3.15}
R_{\sigma} =
\begin{pmatrix}
r_{22} & r_{23} \\
0      & r_{33}
\end{pmatrix},
\end{equation}
and hence,
\begin{equation}
\label{3.16}
\sigma \Psi =
\begin{pmatrix}
 \sigma(\Psi_{22}) & \sigma(\Psi_{23}) \\
 0                 & \sigma(\Psi_{33})
\end{pmatrix} =
\begin{pmatrix}
\Psi_{22} r_{22}                   & \Psi_{22} r_{23} + \Psi_{23} r_{33} \\
0                                  & \Psi_{33} r_{33}
\end{pmatrix}.
\end{equation}
We already know that
$$  r_{22} = c, \quad  r_{33} = c^{-1}.$$

Next, making use of the fundamental relation of the differential
Galois theory $ \sigma \circ \partial = \partial \circ \sigma$, we
get
\begin{align*}
\sigma (\Psi_{23}) &= \sigma (\Psi_{22} \int \left( \frac{i \alpha}{M_2 ^2} + 3 a_3 \right) I^{-2} d t) = \sigma (\Psi_{22})\sigma (\int \left( \frac{i \alpha}{M_2 ^2} + 3 a_3 \right) I^{-2} d t)\\
                   &= c \Psi_{22}  (\int \left( \frac{i \alpha}{M_2 ^2} + 3 a_3 \right) \sigma( I^{-2}) d t     + \nu_2)  \\
                   &=  c^{-1} \Psi_{22} \int \left( \frac{i \alpha}{M_2 ^2} + 3 a_3 \right) I^{-2} d t) + c \nu_2 \Psi_{22} \\
                   &=  c^{-1} \Psi_{23} +  c \nu_2 \Psi_{22}.
\end{align*}
Hence, $r_{23} =  c \nu_2, \nu_2 \in \mathbb{C}$. Moreover, $\nu_2$ is nonzero due to Proposition \ref{prop2}.

Having all the entries of $R_{\sigma}$, we obtain the matrix representation of the Galois group $G = Gal(\mathbb{F}_2/\mathbb{K})$
of the system (\ref{3.5}). This group  is connected $G = G^0$ and solvable, namely
\begin{equation}
\label{3.17}
G^0 = \Bigg\{\begin{pmatrix}
c & c \nu_2 \\
0 & c^{-1}
\end{pmatrix}, c \neq 0, 1, \nu_2 \in \mathbb{C}^* \Bigg\}.
\end{equation}
As can be seen, $G^0$ is  not abelian.

Now, our result follows immediately from  the Morales - Ramis theorem (Theorem \ref{thMR}).

\vspace{2ex}

{\bf Acknowledgements}.
\noindent
My sincerest thanks go to Dr. G. Georgiev for his valuable
remarks and suggestions, which have helped me to considerably
improve the manuscript.
 This work is partially supported by grant KP-06-RILA/7  of Bulgarian National Science Fund.

\end{document}